\documentclass[11pt]{article}
\usepackage[utf8]{inputenc}
\usepackage[margin=1in]{geometry}

\usepackage[rgb,table,usenames,dvipsnames]{xcolor}
\usepackage{pdfpages}

\usepackage{setspace}
\usepackage{verbatim}
\usepackage{subcaption}
\usepackage{boxedminipage}
\usepackage{csquotes}
\usepackage{graphicx}
\usepackage{siamproceedings}
\usepackage{hdb_macros_new}
\usepackage{fullpage}
\usepackage{algorithm}
\usepackage{algpseudocode}
\usepackage{braket}
\usepackage{stmaryrd}
\usepackage{fancyhdr}
\fancypagestyle{plain}{
  \fancyhf{}
  
  \fancyfoot[R]{\scriptsize{Copyright \textcopyright\ 2026\\Copyright for this paper is retained by authors}}
}
\fancypagestyle{empty}{%
  \fancyhf{}%
  \fancyfoot[R]{\scriptsize{Copyright \textcopyright\ 2026\\Copyright for this paper is retained by authors}}%
}

\usepackage[multiuser,inline,nomargin]{fixme}
\FXRegisterAuthor{s}{es}{\color{red}Surendra}
\fxusetheme{color}

\usepackage{xcolor}
\hypersetup{
	colorlinks,
	linkcolor={red!75!black},
	citecolor={blue!75!black},
	urlcolor={red!75!black},
}
 \usepackage[hyperpageref]{backref}

\usepackage{esvect}

\title{Efficient derandomization of differentially private counting queries}
\author{Surendra Ghentiyala \thanks{Cornell University.  \email{sg974@cornell.edu}. This work is supported in part by the NSF under Grants Nos.~CCF-2122230 and CCF-2312296, a Packard Foundation Fellowship, and a generous gift from Google.}}
\date{}

\begin{document}
\pagenumbering{roman}

\maketitle
\listoffixmes

\begin{abstract}
    Differential privacy for the 2020 census required an estimated 90 terabytes of randomness \cite{garfinkel2020randomness}, an amount which may be prohibitively expensive or entirely infeasible to generate. Motivated by these practical concerns, \cite{canonne2025randomness} initiated the study of the randomness complexity of differential privacy, and in particular, the randomness complexity of $d$ counting queries. This is the task of outputting the number of entries in a dataset that satisfy predicates $\mathcal{P}_1, \dots, \mathcal{P}_d$ respectively. They showed the rather surprising fact that though any reasonably accurate, $\varepsilon$-differentially private mechanism for one counting query requires $1-O(\varepsilon)$ bits of randomness in expectation, there exists a fairly accurate mechanism for $d$ counting queries which requires only $O(\log d)$ bits of randomness in expectation.

The mechanism of \cite{canonne2025randomness} is inefficient (not polynomial time) and relies on a combinatorial object known as rounding schemes. Here, we give a polynomial time mechanism which achieves nearly the same randomness complexity versus accuracy tradeoff as that of \cite{canonne2025randomness}. Our construction is based on the following simple observation: after a randomized shift of the answer to each counting query, the answer to many counting queries remains the same regardless of whether we add noise to that coordinate or not. This allows us to forgo the step of adding noise to the result of many counting queries. Our mechanism does not make use of rounding schemes. Therefore, it provides a different---and, in our opinion, clearer---insight into the origins of the randomness savings that can be obtained by batching $d$ counting queries.
\end{abstract}

\thispagestyle{empty}
\newpage
\pagenumbering{arabic}

\section{Introduction}
\label{sec:intro}
Differential privacy is the study of mechanisms that release statistical information about a dataset in a private manner. Informally, a mechanism $\mathcal{M}$ which takes as input a dataset and outputs a statistic, is differentially private, if the following holds. For any individual $i$, the output of $\mathcal{M}$ does not look significantly different in the case when $i$ is in the dataset versus when $i$ is not in the dataset (see \cref{def: differential_privacy}). Since any adversarial data analyst only sees the output of $\mathcal{M}$, any decision that the data analyst makes is essentially the same regardless of whether $i$ is in the dataset or not. Therefore, $i$ can be relatively certain that their membership in the dataset will not have adverse effects on themselves or others. Differential privacy has proven to be an extremely successful framework for private data analysis and machine learning \cite{dwork2006calibrating, dwork2010boosting, dankar2013practicing, abadi2016deep}.

Separately, complexity theory acknowledges that true randomness is both hard to generate and useful for algorithmic tasks \cite{motwani1996randomized}. As a result, complexity theorists view randomness as a valuable resource---much like time or space. This has led to a rich theory of randomness complexity which deals with both how to obtain true randomness from imperfect sources \cite{von1963various, trevisan2001extractors, chattopadhyay2016explicit} as well as derandomization \cite{nisan1994hardness, agrawal2004primes, TCS-010}. We take derandomization here to mean designing an algorithm which performs the same task as the randomized algorithm but uses fewer (or no) random bits.

At the intersection of these two areas, one can and should ask about the possibility of derandomizing differential privacy. The need for randomness efficient implementations of differentially private algorithms is perhaps best illustrated by \cite{garfinkel2020randomness}, which estimated that differential privacy for the 2020 census required 90 \emph{terabytes} of randomness. Motivated by both practical need and natural theoretical questions, \cite{canonne2025randomness} initiated the study of the randomness complexity of differential privacy.

\cite{canonne2025randomness} is concerned with derandomizing counting queries. Say we have a dataset of people $x \in \mathcal{X}^*$ and some predicates $\mathcal{P}_1, \dots, \mathcal{P}_d: \mathcal{X} \rightarrow \{ 0, 1\}$. For example, $\mathcal{P}_1$ may be the predicate ``this person is left-handed''. We wish to determine how many people satisfy $\mathcal{P}_1, \dots, \mathcal{P}_d$ respectively. \cite{canonne2025randomness} showed that if $d=1$, then any $\varepsilon$-differentially private and somewhat accurate mechanism requires $1-O(\varepsilon)$ bits of randomness in expectation (see lemma A.1 of \cite{canonne2025randomness} for a formal statement). One might therefore expect that $d$ counting queries would require $O(d)$ bits of randomness in expectation. Rather surprisingly, \cite{canonne2025randomness} showed that this is \emph{not} the case. One can design a mechanism which is fairly private and accurate (see \cref{subsec: results}) that requires $O(\log d)$ bits of randomness in expectation.

Unfortunately, the mechanism of \cite{canonne2025randomness} is not efficient and they leave it as an open question whether one can design a polynomial time mechanism. Here, we give an efficient derandomization of $d$ counting queries. Furthermore, unlike \cite{canonne2020discrete}, our construction does not make use of the combinatorial objects known as rounding schemes \cite{GeometryOfRounding, geometryOfRoundingOptimal}. We therefore believe that our scheme is more intuitive in many respects. For example, the source of the randomness complexity versus accuracy trade-off is clearer in our scheme (see \cref{subsec: technique}). However, we note that the connection to rounding schemes is still the only way that we know of to obtain lower bounds for the randomness complexity of $d$ counting queries \cite{canonne2025randomness}, a problem we do not address in this work.

\subsection{Our results}
\label{subsec: results}
At the cost of some redundancy, we will present our results for approximate differential privacy and pure differential privacy separately. Although there is significant overlap between these cases, we have split them up for two reasons. First, we believe that since the case for approximate differential privacy is simpler, it serves as a nice warm-up for the case of pure differential privacy. Second, we believe that working with too much generality may obscure our main insight. 

In the following, $R_0(\mathcal{M})$ denotes the maximum random complexity of the mechanism $\mathcal{M}$ and $R(\mathcal{M})$ denotes the expected randomness complexity over the coin tosses of $\mathcal{M}$ (see \cref{subsec: randomness_complexity}). See \cref{sec:prelims} for definitions of $(\varepsilon, \delta)$-differential privacy and $(\alpha, \beta)$-accuracy.
\subsubsection{Approximate differential privacy}
We begin by restating the result of \cite{canonne2025randomness} in the case of approximate differential privacy and then our theorem.
\begin{theorem}[\cite{canonne2025randomness}]
    For every \( d, \ell \in \mathbb{N} \), \( \varepsilon > 0 \), and \( \delta \in \left(0, \frac{1}{d}\right) \), there exists a mechanism $\mathcal{M} : (\{0,1\}^d)^* \rightarrow \mathbb{R}^d$
    that is \((\varepsilon, \delta)\)-differentially private with respect to \( d_{\text{ID}} \), satisfying
    \[
    R_0(\mathcal{M}) \leq \left\lceil \frac{d}{\ell} \right\rceil \cdot \log_2(\ell + 1) + \log_2\left(\frac{1}{\delta}\right) + O(1),
    \]
    such that, for every \( \beta \in (0,1] \), \( \mathcal{M} \) is \((\alpha, \beta)\)-accurate for summation, where
    \[
    \alpha = O\left( \frac{\sqrt{d} \cdot \log(1/\delta) \cdot \log(d/\beta)}{\varepsilon} + \frac{\ell \cdot \sqrt{d} \cdot \log(1/\delta)}{\varepsilon} \right).
    \]
    In particular, taking \( \ell = d \), we obtain $
    R_0(\mathcal{M}) \leq \log_2 d + \log_2\left(\frac{1}{\delta}\right) + O(1)$,
    and
    \[
    \alpha = O\left( \frac{\sqrt{d} \cdot \log(1/\delta) \cdot \log(d/\beta)}{\varepsilon} + \frac{d^{3/2} \cdot \log(1/\delta)}{\varepsilon} \right).
    \]
\end{theorem}

\begin{restatable*}{theorem}{ApproxFinal}
    \label{thm: approx_bigtm}
    Let \( d, s \in \mathbb{N} \), \( \varepsilon > 0 \), and $\delta \leq e^{-\varepsilon/2}$. There exists a mechanism $\mathcal{M} : (\{0,1\}^d)^* \rightarrow \mathbb{R}^d$ running in expected time $\poly \left( n, d, \log s, \log \frac{1}{\varepsilon}, \log \log \frac{1}{\delta} \right)$ that is \((\varepsilon, \delta)\)-differentially private with respect to \( d_{\text{ID}} \), satisfying
    \[
    R(\mathcal{M}) \leq O \left( \frac{d \left( \polylog \frac{d}{\varepsilon} + \polylog \log \frac{1}{\delta} \right)}{s} + \log(s) \right),
    \]
    is \((\alpha, 0)\)-accurate for summation, where
    \[
    \alpha = O \left( \frac{\sqrt{d \ln(d) \cdot \ln \frac{1}{\delta}} \cdot s}{\varepsilon} \right).
    \]
    In particular, taking \( s = d (\polylog \frac{d}{\varepsilon} + \polylog \log \frac{1}{\delta}) \), we obtain a running time of $\poly(n, d, \log \frac{1}{\varepsilon}, \log \log \frac{1}{\delta})$, $
    R(\mathcal{M}) \leq O(\log d + \log(\polylog \frac{d}{\varepsilon} + \polylog \log \frac{1}{\delta}))$,
    and
    \[
    \alpha = O \left( \frac{d^{3/2} \sqrt{ \ln(d) \cdot \ln \frac{1}{\delta}} \cdot \left( \polylog \frac{d}{\varepsilon} + \polylog \log \frac{1}{\delta} \right)}{\varepsilon} \right).
    \]
\end{restatable*}

There are several differences between the above results. The two biggest are that our mechanism is efficient, but it only has a bound for the expected randomness usage, not a worst case randomness complexity bound. 

We now address a few other differences. Observe that the result of \cite{canonne2025randomness} is particularly fine-grained in terms of its randomness complexity; it only includes an additive $O(1)$ factor. The randomness complexity of our mechanism is essentially the same, except for a multiplicative $\polylog \frac{d}{\varepsilon} + \polylog \log \frac{1}{\delta}$ factor, which we generally expect to be quite small. In particular, in the extremal case of $s = d (\polylog \frac{d}{\varepsilon} + \polylog \log \frac{1}{\delta})$, we obtain a randomness complexity of $O(\log d + \log(\polylog \frac{d}{\varepsilon} + \polylog \log \frac{1}{\delta}))$, which we expect to be $O(\log d)$, other than for minuscule values of $\varepsilon, \delta$. When it comes to accuracy, our algorithm is always accurate and we therefore have no need for this $\beta$ value. Up to an additional small factor of $\sqrt{\ln(d)} \cdot \left( \polylog \frac{d}{\varepsilon} + \polylog \log \frac{1}{\delta} \right)$, the accuracy of our mechanism matches that of \cite{canonne2025randomness}.

\subsubsection{Pure differential privacy}
We now compare the theorem of \cite{canonne2025randomness} to our theorem in the case of pure differential privacy.
\begin{theorem}[\cite{canonne2025randomness}]
    For every $d, \ell, \varepsilon > 0$, there is a mechanism $
    \mathcal{M} : (\{0,1\}^d)^* \to \mathbb{R}^d$ that is $\varepsilon$-DP with respect to $\mathrm{d}_{\mathrm{ID}}$ satisfying
    \[
    R(\mathcal{M}) \leq \left\lceil \frac{d}{\ell} \right\rceil \cdot \log_2(\ell + 1) + O(1),
    \]
    such that, for every $\beta \in (0,1]$, $\mathcal{M}$ is $(\alpha, \beta)$-accurate for summation for
    \[
    \alpha = O\left(
    \frac{d \cdot \log(d/\beta)}{\varepsilon} + \frac{\ell \cdot d \cdot \log d}{\varepsilon}
    \right).
    \]
    In particular, taking $\ell = d$, we obtain $R(\mathcal{M}) \leq \log_2 d + O(1)$
    and
    \[
    \alpha = O\left(
    \frac{d \cdot \log(d/\beta)}{\varepsilon} + \frac{d^2 \cdot \log d}{\varepsilon}
    \right).
    \]
\end{theorem}

\begin{restatable*}{theorem}{PureFinal}
    \label{thm: pure_bigtm}
    Let \( d, s \in \mathbb{N} \), and \( \varepsilon > 0 \). There exists a mechanism $\mathcal{M} : (\{0,1\}^d)^* \rightarrow \mathbb{R}^d$ running in expected time $\poly \left( n, d, \log s, \log \frac{1}{\varepsilon} \right)$ that is \((\varepsilon, 0)\)-differentially private with respect to \( d_{\text{ID}} \), satisfying
    \[
    R(\mathcal{M}) \leq O \left( \frac{d}{s} \cdot \polylog \frac{d}{\varepsilon} + \log s \right),
    \]
    is \((\alpha, \beta)\)-accurate for summation, where
    \[
    \alpha = O\left( \frac{d \cdot \log(d/\beta)}{\varepsilon} + \frac{s \ln(s) \cdot d \ln(\frac{d}{\varepsilon})}{\varepsilon} \right).
    \]
    In particular, taking \( s = d \polylog \frac{d}{\varepsilon} \), we obtain a running time of $\poly(n, d, \log \frac{1}{\varepsilon})$, $
    R(\mathcal{M}) \leq O(\log d + \log \polylog \frac{d}{\varepsilon})$,
    and
    \[
    \alpha = O\left( \frac{d \cdot \log(d/\beta)}{\varepsilon} + \frac{d^2 \cdot \polylog(\frac{d}{\varepsilon})}{\varepsilon} \right).
    \]
\end{restatable*}
We observe that the major differences between the above are that our mechanism is efficient while the result of \cite{canonne2025randomness} is particularly fine-grained. Our mechanism has essentially the same randomness complexity as that of \cite{canonne2020discrete} up to a multiplicative $\polylog \frac{d}{\varepsilon}$ factor. The randomness complexity we obtain in the extremal case of $s = d \polylog \frac{d}{\varepsilon}$ is $O(\log d + \log \polylog \frac{d}{\varepsilon})$, which we generally expect to be $O(\log d)$ except for truly minuscule values of $\varepsilon$. Finally, in terms of accuracy, we obtain the same result as \cite{canonne2025randomness} up to $\polylog \frac{d}{\varepsilon}$ factors (we assume $s \leq d \log \frac{d}{\varepsilon}$ since taking $s$ to be any larger has no asymptotically appreciable effect on randomness complexity), which we do not expect to be particularly large.

\subsection{Technique}
\label{subsec: technique}
Our mechanism is best understood as being built through a series of iterative refinements. We now describe how these refinements work in the case of approximate differential privacy (\cref{sec: approx}). The case for pure differential privacy (\cref{sec: pure}) is rather similar but requires more care in handling edge cases.

Suppose we have some dataset and properties $\mathcal{P}_1, \dots, \mathcal{P}_d$. Let $x \in (\{ 0, 1\}^d)^*$ be the vector such that ${x_i}_j = 1$ if and only if person $i$ satisfies property $j$. We would like to release $v \in \mathbb{N}^d$ where $v = \text{sum}(x)$. Of course, releasing $v$ would completely violate differential privacy. Instead consider the mechanism $\mathcal{M}_1$ which outputs $v + (\eta_1, \dots, \eta_d)$ where $\eta_i \sim \mathcal{N}_{\mathbb{Z}}(\sigma^2)$, the discrete Gaussian with appropriately chosen parameter $\sigma^2$. $\mathcal{M}_1$ is the well-known discrete Gaussian mechanism \cite{canonne2020discrete}, which is known to be approximately differentially private. Unfortunately, sampling each $\eta_i$ requires $\polylog( \sigma)$ bits of randomness. Therefore, sampling $(\eta_1, \dots, \eta_d)$ requires $O(d \cdot \polylog (\sigma))$ bits of randomness.

For our second refinement $\mathcal{M}_2$, let $r \in \mathbb{N}$ be an integer such that for $Y \sim \mathcal{N}_{\mathbb{Z}}(\sigma^2)$, with overwhelming probability, $|Y| \leq r$. Let $\mathcal{M}_2$ be the mechanism where we output $v + (\eta_1, \dots, \eta_d)$  where each $\eta_i$ is sampled from $\mathcal{N}_{\mathbb{Z}}(\sigma^2)$ conditioned on $|\eta_i| < r$. Rejection sampling bounds tell us that we are only going to need to sample each $\eta_i$ from $\mathcal{N}_{\mathbb{Z}}(\sigma^2)$ $O(1)$ times to ensure $|\eta_i| < r$. Therefore, sampling each $\eta_i$ remains an operation that only requires $\polylog(\sigma)$ bits in expectation. Crucially, $\mathcal{M}_2$ is still differentially private. Intuitively, this is because $\mathcal{M}_1$ was differentially private, and $\mathcal{M}_2$ is almost the same as the previous mechanism (formally, the statistical distance between the two is incredibly small). Therefore, if all adversaries $\mathcal{A}$ had difficulty distinguishing between the $\mathcal{M}_1(x), \mathcal{M}_1(x')$ before, $\mathcal{A}$ should have a hard time distinguishing between $\mathcal{M}_2(x), \mathcal{M}_2(x')$ since the statistical distance $\Delta(\mathcal{M}_1(y), \mathcal{M}_2(y))$ is very small for all $y$.

For our third refinement $\mathcal{M}_3$, let us fix some integer $s$ between $1$ and $d$. $\mathcal{M}_3$ samples $\omega$ uniformly from $\{ r, 2r, \dots, rs \}$ and then outputs $(\lfloor v_1 + \eta_1 + \omega \rfloor_{rs},  \dots, \lfloor v_d + \eta_d + \omega \rfloor_{rs})$ where each $\eta_i$ is sampled from $\mathcal{N}_{\mathbb{Z}}(\sigma^2)$ conditioned on $|\eta_i| < r$. Here, $\lfloor \cdot \rfloor_{rs}$ denotes rounding down to the nearest multiple of $rs$. $\mathcal{M}_3$ is differentially private by the well-known postprocessing property of differential privacy (\cref{lem: postprocessing}) which says that if $\mathcal{M}_2$ is differentially private, then $f \circ \mathcal{M}_2$ is differentially private for any randomized function $f$. In particular, $\mathcal{M}_3$ has the same distribution as taking the output of $\mathcal{M}_2$, call it $z$, sampling $\omega$ uniformly from $\{ r, 2r, \dots, rs \}$, and outputting $(\lfloor z_1 + \omega \rfloor_{rs}, \dots, \lfloor z_d + \omega \rfloor_{rs})$.

Now we are ready to make the key observation. Consider any coordinate $i \in [d]$. In $\mathcal{M}_3$, with high probability, $\lfloor v_i + \eta_i + \omega \rfloor_{rs}$ does not depend on $\eta_i$. In particular, with at least $1-2/s$ probability over our choice of $\omega$, the following holds.
\[ \lfloor v_i + \omega -r \rfloor_{rs} = \lfloor v_i + \omega +r \rfloor_{rs} \]
Therefore, regardless of what $\eta_i$ is, $\mathcal{M}_3$ outputs the same thing for coordinate $i$ since $\eta_i$ is guaranteed to be between $-r$ and $r$. This is where the savings in randomness complexity come from. We will design a mechanism where once we sample $\omega$, we will only sample $\eta_i$ if $\lfloor v_i + \omega + \eta_i \rfloor_{rs}$ is ambiguous. In expectation, only $2d/s$ coordinates $i$ will have have an ambiguous $\lfloor v_i + \omega -r \rfloor_{rs}$ (by linearity of expectation), so we will only need to sample $\eta_i$ $2d/s$ times in expectation. Notice that the sampling of $\omega$ requires $O(\log s)$ bits of randomness. Furthermore, the process we have outlined is computationally efficient.

The loss in accuracy comes from the rounding to $rs$. It should now be clear where the tradeoff between randomness complexity and accuracy comes from. As we increase $s$, we need to sample fewer $\eta_i$ (only $2d/s$), but we are forced to round to a coarser grid (one with side lengths $rs$).

\section{Preliminaries}
\label{sec:prelims}
\snote{Say that you want $d/\varepsilon$ to be an integer to simplify things.}
We will use the notation $\lfloor \cdot \rfloor_k$ to denote rounding down to a multiple of $k$. If the input to $\lfloor \cdot \rfloor_k$ is a vector, each coordinate is rounded down to a multiple of $k$. We use the notation $U[S]$ to denote a uniformly sampled value from $S$.
\subsection{Differential privacy}
We will define the requisite terms only for the case of differentially private vector sum. For more general definitions, see \cite{dwork2014algorithmic, canonne2025randomness}. Any dataset $x \subseteq \mathcal{X}^n$ is a set of $n$ entries from some domain $\mathcal{X}$. At a high level, we say that a mechanism is differentially private if its output is close in statistical distance on $x, x'$ for any two datasets $x, x'$ which differ on exactly one entry. The notion of ``differing on one entry'' is underspecified. So we first define the $d_{\mathrm{ID}}$ database distance and then formally define differential privacy. For succinctness, the database metric is always assumed to be $d_{\mathrm{ID}}$ unless otherwise specified.

\begin{definition}[$d_{\mathrm{ID}}$ \cite{canonne2025randomness}]
    \label{def: database_dist}
    For two databases \( x, x' \in \mathcal{X}^* \), their \emph{insert-delete distance} is the minimum number of insertions and deletions of elements of \( \mathcal{X} \) needed to transform \( x \) into \( x' \).
\end{definition}

\begin{definition}[differential privacy \cite{canonne2025randomness}]
    \label{def: differential_privacy}
    Fix \( \varepsilon > 0 \) and \( \delta \in [0, 1] \).
    A randomized algorithm $\mathcal{M} : \mathcal{X}^* \to Y$
    is \((\varepsilon, \delta)\)-differentially private (or \((\varepsilon, \delta)\)-DP) with respect to database metric \( d \) if for every pair of adjacent databases \( x \sim_d x' \) in \( \mathcal{X}^* \) and every measurable \( S \subseteq Y \), we have
    \[
    \Pr[ \mathcal{M}(x) \in S ] \leq e^{\varepsilon} \cdot \Pr[ \mathcal{M}(x') \in S ] + \delta.
    \]
    If \( \delta = 0 \), we simply say that \( \mathcal{M} \) is \( \varepsilon \)-DP.
\end{definition}

Just as in \cite{canonne2025randomness}, we will restrict our attention to the case where $\mathcal{X} = \{ 0, 1\}^d$ and a single entry from $\mathcal{X}$ tells us if an individual satisfies properties $\mathcal{P}_1, \dots, \mathcal{P}_d$ respectively. Of course, the trivial mechanism $\mathcal{M}(x) = 0$ satisfies privacy but is not at all accurate. We want our mechanism to satisfy the following accuracy property.

\begin{definition}[accuracy \cite{canonne2025randomness}]
    \label{def: accuracy}
    Given \( \alpha \geq 0 \) and \( \beta \in [0, 1] \), a randomized algorithm $\mathcal{M} : (\{0, 1\}^d)^n \to \mathbb{R}^d$ is said to be \((\alpha, \beta)\)-accurate for summation if for every \( x \in (\{0, 1\}^d)^* \), we have
    \[
    \Pr\left[ \| \mathcal{M}(x) - \mathrm{sum}(x) \|_\infty > \alpha \right] \leq \beta,
    \]
    where
    \[
    \mathrm{sum}(x) = \sum_{i=1}^{|x|} x_i \in \mathbb{R}^d.
    \]
\end{definition}

One of the properties of differential privacy that will be key to our mechanisms is postprocessing, which says that applying any (possibly randomized) function $f$ to the output of a differentially private algorithm does not harm differential privacy.
\begin{lemma}[Postprocessing]
    \label{lem: postprocessing}
    Let $M: \mathcal{X} \rightarrow \mathcal{Y}$ be an $(\varepsilon, \delta)$ differentially private mechanism and $f: \mathcal{Y} \rightarrow \mathcal{Z}$ be a randomized function. $f \circ M: \mathcal{X} \rightarrow \mathcal{Z}$ is $(\varepsilon, \delta)$ differentially private.
\end{lemma}

\subsection{Randomness complexity}
\cite{canonne2025randomness} defined the randomness complexity of a mechanism. We note that the randomness complexity of all our mechanisms is stated in terms of $R(\mathcal{M})$, which is the \emph{expected} number of random bits used by $\mathcal{M}$ over the coin tosses of $\mathcal{M}$. We only define $R_0(\mathcal{M})$ for comparison with \cite{canonne2025randomness} in \cref{subsec: results}.
\label{subsec: randomness_complexity}
\begin{definition}[\cite{canonne2025randomness}]
    For a randomized algorithm \( \mathcal{M} : \mathcal{X}^n \to Y \), we define:
\begin{itemize}
    \item \( R(\mathcal{M}) \): the maximum over \( x \in \mathcal{X}^n \) of the expected number of random bits used by \( \mathcal{M} \) on input \( x \), where the expectation is taken over the coin tosses of \( \mathcal{M} \).
    \item \( R_0(\mathcal{M}) \): the maximum over \( x \) of the maximum number of random bits used by \( \mathcal{M} \) on input \( x \).
\end{itemize}

\end{definition}

\subsection{Discrete Gaussian}
For our approximately differentially private mechanism (\cref{sec: approx}), we will make use of discrete Gaussian noise. The discrete Gaussian has the crucial advantage over the continuous gaussian that it is easy to work with and samplable on a finite machine. We first define the discrete gaussian and outline how the discrete Gaussian mechanism is approximately differentially private.
\begin{definition}[$\mathcal{N}_{\mathbb{Z}}(\sigma^2)$]
    The distribution $\mathcal{N}_{\mathbb{Z}}(\sigma^2)$ is supported on $x \in \mathbb{Z}$ and has the following probability mass function.
    \[ \underset{X \sim \mathcal{N}_{\mathbb{Z}}(\sigma^2)}{\mathbb{P}}[X =x] = \frac{e^{-(x^2)/2\sigma^2}}{\sum_{y \in \mathbb{Z}} e^{-(y^2)/2\sigma^2}} \]
\end{definition}

\begin{lemma}[Discrete Gaussian Privacy \cite{canonne2020discrete}]
    \label{lem: gaussian_privacy}
    Let $\sigma, \varepsilon>0$. Let $d/\sigma^2 \leq \rho^2$. Define the randomized mechanism for $d$ counting queries $\mathcal{M}: (\{ 0, 1\}^d)^n \rightarrow \mathbb{Z}^d$ by $\mathcal{M}(x) = \text{sum}(x) + Y$ for $Y \sim \mathcal{N}_{\mathbb{Z}}(0, \sigma^2)^d$. $\mathcal{M}$ satisfies $\frac{1}{2} \rho^2$-concentrated differential privacy and therefore satisfies $(\varepsilon, \delta)$-differential privacy for $\delta = e^{-(\varepsilon-0.5\rho^2)^2/2\rho^2}$.
\end{lemma}

\begin{corollary}
    \label{cor: gaussian_privacy}
    Assume $\delta \leq e^{-\varepsilon/2}$. Let $\sigma^2 = 4d\ln(1/\delta)/\varepsilon^2$. Define the randomized mechanism for $d$ counting queries $\mathcal{M}: (\{ 0, 1\}^d)^n \rightarrow \mathbb{Z}^d$ by $\mathcal{M}(x) = \text{sum}(x) + Y$ for $Y \sim \mathcal{N}_{\mathbb{Z}}(0, \sigma^2)^d$. $\mathcal{M}$ is $(\varepsilon, \delta)$ differentially private.
\end{corollary}
\begin{proof}
    Let $\rho^2 = \varepsilon^2/(4 \ln(1/\delta))$. Then $d/\sigma^2 = \rho^2$. We therefore only need to check $e^{-(\varepsilon-0.5\rho^2)^2/2\rho^2} \leq \delta$ to apply \cref{lem: gaussian_privacy}. It suffices to show the following.
    \begin{align*}
        \frac{(\varepsilon-\frac{1}{2} \rho^2)^2}{2 \rho^2} &\geq \ln(1/\delta)
    \end{align*}

    The left hand side is equal to the following.
    \[ \frac{(\varepsilon-\frac{1}{2} \rho^2)^2}{2 \rho^2}
    = \frac{(\varepsilon-\frac{\varepsilon^2}{8 \ln(1/\delta)})^2}{\frac{\varepsilon^2}{2 \ln(1/\delta)}}
    = 2 \ln(1/\delta) - \frac{\varepsilon}{2} + \frac{\varepsilon^2}{32 \ln(1/\delta)} \]

    To show that this is greater than $\ln(1/\delta)$, it suffices to show the following.
    \[ \ln(1/\delta) - \frac{\varepsilon}{2} + \frac{\varepsilon^2}{32 \ln(1/\delta)} \geq 0\]

    Which is implied by the fact that $\delta \leq e^{-\varepsilon/2}$.
\end{proof}

We will also make use of the following tail bound and randomness complexity bound for sampling from the discrete Gaussian.
\begin{lemma}[Discrete Laplace tail bound \cite{canonne2020discrete}]
    \label{lem: gauss_tail}
    \[ \underset{X \sim \mathcal{N}_{\mathbb{Z}}(0, \sigma^2)}{\mathbb{P}}[X \geq \lambda] \leq e^{-\lambda^2/2\sigma^2}\]
\end{lemma}

\begin{lemma}[Discrete Gaussian sampling complexity \cite{canonne2020discrete}]
    \label{lem: gauss_sample}
    Say $\sigma$ has a $k$ bit binary representation. There exists an algorithm to sample from $\mathcal{N}_{\mathbb{Z}}(0, \sigma^2)$ with expected running time $\poly(k)$ which uses $\poly(k)$ bits of randomness in expectation.
\end{lemma}

\subsection{Discrete Laplace}
For our purely differentially private mechanism (\cref{sec: pure}), we will use the discrete Laplace mechanism. We begin by giving two equivalent definitions of the discrete Laplace distributions as well the differential privacy properties of the discrete Laplace distribution.
\begin{definition}[$\text{Lap}_{\mathbb{Z}}(t)$ \cite{inusah2006discrete}]
    \label{lem: discrete_lap_dist}
    The distribution $\text{Lap}_{\mathbb{Z}}(t)$ is supported on $x \in \mathbb{Z}$ and has the following probability mass function.
    \[ \underset{X \sim \text{Lap}_{\mathbb{Z}}(t)}{\mathbb{P}}[X =x] = \frac{e^{1/t}-1}{e^{1/t}+1} e^{-|x|/t} \]
    Equivalently, $X \sim \text{Lap}_{\mathbb{Z}}(t)$ is identically distributed to $U-V$ where $U, V$ are i.i.d geometric random variables with parameter $1-e^{-1/t}$.
\end{definition}

\begin{lemma}[Discrete Laplace Privacy \cite{canonne2020discrete}]
    \label{lem: discrete_lap_private}
    Let $\varepsilon > 0$. Say $f: \mathcal{X}^n \rightarrow \mathbb{Z}$ has $\ell_1$ sensitivity $\Delta_1(f) = \max_{x \sim_d x'} \| f(x) - f(x') \|_1$. Define a randomized algorithm $M: \mathcal{X}^n \rightarrow \mathbb{Z}$ by $M(x) = q(x) + \text{Lap}_{\mathbb{Z}}(\Delta_1(f)/\varepsilon)$. Then $M$ satisfies $(\varepsilon, 0)$ differential privacy.
\end{lemma}

We will also make use of the following tail bound and randomness complexity bound for sampling from the discrete Laplace distribution.
\begin{lemma}[Discrete Laplace tail bound]
    \label{lem: discrete_lap_tail}
    Let $Y$ be distributed as $\text{Lap}_{\mathbb{Z}}(1/\varepsilon)$. For any $m \in \mathbb{N}$ the following tail bound holds.
    \[ \mathbb{P}[Y \geq m] = \mathbb{P}[Y \leq -m] = \frac{e^{-\varepsilon(m-1)}}{e^\varepsilon + 1} \leq e^{-\varepsilon(m-1)} \]
\end{lemma}

\begin{lemma}[Discrete Laplace sampling complexity \cite{canonne2020discrete}]
    \label{lem: lap_sample}
    Say $t$ has a $k$ bit binary representation. There exists an algorithm to sample from $\text{Lap}_{\mathbb{Z}}(t)$ with expected running time $\poly(k)$ which uses $\poly(k)$ bits of randomness in expectation.
\end{lemma}

\subsection{Sampling}
It will often be useful for us to rejection sample $X$ based on some event happening. The following lemma shows that this is not more significantly time or randomness intensive than just sampling $X$. 
\begin{lemma}
    \label{lem: rejection_sampling}
    Let $X$ be a random variable such that the following holds. There exists an algorithm to sample from $X$ which has expected running time $T$ and requires $K$ bits of randomness in expectation. Let $E$ be an event checkable in time $T$ which happens with probability $p>0$. We can sample from $X$ conditioned on $E$ in expected time $O(T/p)$ using $k/p$ bits in expectation.
\end{lemma}
\begin{proof}
    Let $Y_i$ denote the number of random bits needed to draw the $i^{\text{th}}$ sample of $X$ and let $Y_i = 0$ if the rejection sampling algorithm has succeeded before draw $i$. Let $Y = \sum_{i=1}^\infty Y_i$
    \begin{align*}
        \mathbb{E}[Y] &= \sum_{i=1}^\infty \mathbb{E}[Y_i]
        = \sum_{i=1}^\infty \mathbb{E}[Y_i | Y_i \neq 0] \cdot \mathbb{P}[Y_i \neq 0]
        = \sum_{i=1}^\infty K \cdot (1-p)^{i-1}
        = K/p
    \end{align*}
    The same exact analysis tells us that the expected running time is $O(T/p)$.
\end{proof}

\begin{lemma}[\cite{canonne2020discrete}]\footnote{This is implicit in Algorithm 2 of \cite{canonne2020discrete} where they have an intermediate variable which is distributed as $\text{Ber}(1-e^{-1/t})$.}
    \label{lem: ber_exp_sampling}
    If $\gamma$ has a $k$ bit binary representation, then there exists an algorithm to sample from $\text{Ber}(1-e^{-1/\gamma})$ with expected running time $\poly(k)$ which uses $\poly(k)$ bits of randomness in expectation.
\end{lemma}

\section{Approximate differential privacy}
\label{sec: approx}
We first state our main result for the case of approximate differential privacy. We will then give the mechanism, and prove its privacy, time and randomness complexity, and accuracy separately.
\ApproxFinal
\begin{proof}
    Combine \cref{lem: approx_privacy}, \cref{lem: approx_complexity}, and \cref{lem: approx_accuracy}.
\end{proof}

\subsection{The algorithm}
For the sake of analysis, we give a series of algorithms that all have the desired differential privacy guarantee, though only \cref{alg: approx_mec4} is the one which we will show has low randomness complexity.
\begin{algorithm}
\caption{Mechanism 1}
\label{alg: approx_mec1}
\begin{algorithmic}
\Require $x \in (\{ 0, 1\}^d)^n$, $\sigma$

\State $\eta \sim  \mathcal{N}_{\mathbb{Z}}(0, \sigma^2)^d$

\noindent \Return $y \leftarrow \text{sum}(x)+\eta$

\end{algorithmic}
\end{algorithm}

\begin{algorithm}
\caption{Mechanism 2}
\label{alg: approx_mec2}
\begin{algorithmic}
\Require $x \in (\{ 0, 1\}^d)^n$, $\sigma$, and integer $r > 0$

\State $\eta \sim  \mathcal{N}_{\mathbb{Z}}(0, \sigma^2)^d$ conditioned on $\| \eta \|_{\infty} < r$

\noindent \Return $y \leftarrow \text{sum}(x)+\eta$

\end{algorithmic}
\end{algorithm}

\begin{algorithm}
\caption{Mechanism 3}
\label{alg: approx_mec3}
\begin{algorithmic}
\Require $x \in (\{ 0, 1\}^d)^n$, $\sigma$, and integers  $r, s > 0$

\State $\omega \sim U[\{ 1, \dots, s \}] \cdot r$

\State $\eta \sim  \mathcal{N}_{\mathbb{Z}}(0, \sigma^2)^d$ conditioned on $\| \eta \|_{\infty} < r$

\State $y \leftarrow \text{sum}(x)+ (\omega, \dots, \omega) + \eta$

\noindent \Return $\lfloor y \rfloor_{rs}$

\end{algorithmic}
\end{algorithm}

\snote{In the below, $\lfloor \text{sum(x)}_i + \omega \rfloor_{rs} = \lfloor \text{sum(x)}_i + \omega +r \rfloor_{rs}$ needs to be changed to $\lfloor \text{sum(x)}_i + \omega -r \rfloor_{rs} = \lfloor \text{sum(x)}_i + \omega +r \rfloor_{rs}$}
\begin{algorithm}
\caption{Mechanism 4}
\label{alg: approx_mec4}
\begin{algorithmic}
\Require $x \in (\{ 0, 1\}^d)^n$, $\sigma$, and integers $r, s > 0$

\State $\omega \sim U[\{ 1, \dots, s \}] \cdot r$

\For {$i \in [1, d]$}

\If {$\lfloor \text{sum(x)}_i + \omega - r \rfloor_{rs} = \lfloor \text{sum(x)}_i + \omega +r \rfloor_{rs}$}
    \State $y_i \leftarrow \lfloor \text{sum(x)}_i + \omega \rfloor_{rs}$
\Else
\State Resample $\eta_i \sim  \mathcal{N}_{\mathbb{Z}}(0, \sigma^2)$ until $|\eta_i| < r$

\State $y \leftarrow \lfloor \text{sum}(x)_i+ \omega + \eta_i \rfloor_{rs}$
\EndIf
\EndFor

\noindent \Return $(y_1, \dots, y_d)$

\end{algorithmic}
\end{algorithm}

\begin{lemma}
    \label{lem: same_dist_approx}
    For all $\sigma, r, s$, \cref{alg: approx_mec3} and \cref{alg: approx_mec4} have the same output distribution.
\end{lemma}
\begin{proof}
    It suffices to show that conditioned on $\omega$, the output distributions of $(y_1, \dots, y_d)$ are the same. Fix $\omega$ and consider the distribution of outputs of \cref{alg: approx_mec3} and \cref{alg: approx_mec4} conditioned on $ \omega$, call these $\mathcal{D}_3, \mathcal{D}_4$. Notice that in the output of both $\mathcal{D}_3$ and $\mathcal{D}_4$, each coordinate is distributed independently, so we simply need to show that the distributions of coordinate $i$ of $\mathcal{D}_3$ and coordinate $i$ of $\mathcal{D}_4$, call these $X_i, Y_i$, have the same distribution. If $\lfloor \text{sum}(x)_i + \omega -r \rfloor_{rs} = \lfloor \text{sum}(x)_i + \omega + r \rfloor_{rs}$, then the distributions of $X_i$ and $Y_i$ are clearly the same since $X_i = \lfloor \text{sum}(x)_i + \omega \rfloor_{rs} = Y_i$.  If $\lfloor \text{sum}(x)_i + \omega -r \rfloor_{rs} \neq \lfloor \text{sum}(x)_i + \omega + r \rfloor_{rs}$, then the distributions $X_i, Y_i$ are clearly the same by definition.
\end{proof}

\subsection{Privacy}

To show that \cref{alg: approx_mec4} is private, we first show the following lemma, which states that small statistical shifts to the output distribution of a mechanism do not significantly affect privacy.
\begin{lemma}
    \label{lem: stat_dist}
    Let $\mathcal{M}$ and $\mathcal{M}'$ be mechanisms such that for all inputs $x$, the statistical distance between $\mathcal{M}(x)$ and $\mathcal{M}(x')$ is at most $\gamma$. If $\mathcal{M}$ is $(\varepsilon, \delta)$-DP, then $\mathcal{M}'$ is $(\varepsilon, (e^{\varepsilon}+1)\gamma + \delta)$-DP.
\end{lemma}
\begin{proof}
    Let $S$ be any measurable set and $x$ and $x'$ be any two neighboring datasets.
    \begin{align*}
        \mathbb{P}[\mathcal{M}'(x) \in S] &\leq \mathbb{P}[\mathcal{M}(x) \in S] + \gamma\\
        &\leq e^{\varepsilon} \cdot \mathbb{P}[\mathcal{M}(x') \in S] + \delta + \gamma\\
        &= e^{\varepsilon} \cdot (\mathbb{P}[\mathcal{M}'(x') \in S]+\gamma) + \delta + \gamma\\
        &= e^{\varepsilon} \cdot \mathbb{P}[\mathcal{M}'(x') \in S]+ (e^\varepsilon+1)\gamma + \delta \\
    \end{align*}
\end{proof}

\begin{lemma}
    \label{lem: approx_privacy}
    Assume $\delta \leq e^{-\varepsilon/2}$. Let $\gamma = \frac{1}{2} \delta/(e^\varepsilon + 1)$, $\sigma^2 = 4d\ln(2/\delta)/\varepsilon^2$ and $r = \sigma \sqrt{2 \ln(d) \cdot \ln (1/\gamma)}$. Then \cref{alg: approx_mec4} is $(\varepsilon, \delta)$ differentially private.
\end{lemma}
\begin{proof}
    We will show that since \cref{alg: approx_mec1} is $(\varepsilon, \delta/2)$-differentially private, \cref{alg: approx_mec4} is $(\varepsilon, \delta)$-differentially private.
    \begin{enumerate}
        \item \cref{alg: approx_mec1} is $(\varepsilon, \delta/2)$-differentially private by \cref{cor: gaussian_privacy}.

        \item Begin by considering \cref{alg: approx_mec1}. Each $\eta_i$ was greater than or equal to $r$ with probability at most $\gamma/d$ by \cref{lem: gauss_tail}. By the union bound, $\| \eta \|_\infty \geq r$ with probability at most $\gamma$. Therefore, the statistical distance between the output of \cref{alg: approx_mec1} and \cref{alg: approx_mec2} is at most $\gamma$. \cref{lem: stat_dist} now tells us that since \cref{alg: approx_mec1} is $(\varepsilon, \delta/2)$ secure, \cref{alg: approx_mec2} is $(\varepsilon, (e^\varepsilon + 1) \gamma + \delta/2)$-differentially private. By our choice of $\gamma$, this means \cref{alg: approx_mec2} is $(\varepsilon, \delta)$-differentially private.

        \item Notice that \cref{alg: approx_mec3} can be seen as taking the output of \cref{alg: approx_mec2}, adding $(\omega, \dots, \omega)$ to it where $\omega \sim U[\{ 1, \dots, s\}] \cdot r$, and then rounding each coordinate down to a multiple of $rs$. Therefore, \cref{lem: postprocessing} and the fact that \cref{alg: approx_mec2} is $(\varepsilon, \delta)$-differentially private tell us \cref{alg: approx_mec3} is $(\varepsilon, \delta)$-differentially private.

        \item Since \cref{alg: approx_mec3} is $(\varepsilon, \delta)$-differentially private and by \cref{lem: same_dist_approx}, \cref{alg: approx_mec3} and \cref{alg: approx_mec4} have the same exact output distribution, \cref{alg: approx_mec4} is $(\varepsilon, \delta)$-differentially private.
    \end{enumerate}
\end{proof}

\subsection{Time and randomness complexity}
We first prove the following lemma, which will let us deduce that in \cref{alg: approx_mec4}, for any $i$, the probability that $\lfloor \text{sum}(x)_i + \omega -r \rfloor_{rs} = \lfloor \text{sum}(x)_i + \omega +r \rfloor_{rs}$ is high. This will prove useful later, since if that conditional holds, then no bits of randomness are used in loop iteration $i$.
\begin{lemma}
    \label{lem: 1/s_prob}
    Let $v \sim U[\{1, \dots, s \}]$ and $y \in \mathbb{Z}^+$. The probability that $\lfloor y + rv - r \rfloor_{rs} \neq \lfloor y + rv + r \rfloor_{rs}$ at most $2/s$.
\end{lemma}
\begin{proof}
    Scaling everything down by $r$, the question becomes the following. Let $v \sim U[\{1, \dots, s \}]$ and $y \in \mathbb{Z}^+$, what is the probability that $\lfloor y/r + v -1 \rfloor_{s} \neq \lfloor y/r + v + 1 \rfloor_{s}$? Notice that $\lfloor y/r + q \rfloor_{s}$ is non-decreasing and takes on at most 3 values for $q \in \{ 0, \dots, s+1 \}$. Therefore, there can be only 2 $v$ such that $\lfloor y/r + v - 1 \rfloor_{s} \neq \lfloor y/r + v + 1 \rfloor_{s}$, so the probability of sampling such $v$ is at most $2/s$.
\end{proof}

\begin{lemma}
    \label{lem: approx_complexity}
    Assume $\delta \leq e^{-\varepsilon/2}$. Let $\gamma = \frac{1}{2} \delta/(e^\varepsilon + 1)$, $\sigma^2 = 4d\ln(2/\delta)/\varepsilon^2$ and $r = \sigma \sqrt{2 \ln(d) \cdot \ln (1/\gamma)}$. Let $\mathcal{M}$ be \cref{alg: approx_mec4}. $\mathcal{M}$ runs in expected time $\poly \left( n, d, \log s, \log \frac{1}{\varepsilon}, \log \log \frac{1}{\delta} \right)$ and
    \[ R(\mathcal{M}) =O \left( \frac{d \left( \polylog \frac{d}{\varepsilon} + \polylog \log \frac{1}{\delta} \right)}{s} + \log(s) \right) \ .\]
\end{lemma}
\begin{proof}
    $ $\newline
    \textbf{Randomness complexity:} Let $X_i$ denote the number of random bits used by iteration $i$ of the loop. Let $E_i$ denote the event that $\lfloor \text{sum(x)}_i + \omega -r  \rfloor_{rs} = \lfloor \text{sum(x)}_i + \omega +r \rfloor_{rs}$.

    \[ \mathbb{E}[X_i] = \mathbb{E}[X_i | E_i] \cdot \mathbb{P}[E_i] + \mathbb{E}[X_i | \overline{E}_i] \cdot \mathbb{P}[\overline{E}_i]\]

    By \cref{lem: 1/s_prob}, $\mathbb{P}[E_i] \geq (1-2/s)$. Furthermore, if $E_i$ occurs, $X_i = 0$ since the loop uses no randomness. Therefore, the above simplifies to the following.
    \[ \mathbb{E}[X_i] \leq \frac{2 \mathbb{E}[X_i | \overline{E}_1]}{s} \]

    $\mathbb{E}[X_i | \overline{E}_1]$ is the number of bits required to sample from $\eta_i$ conditioned on $|\eta_i| < r$ using rejection sampling. By our choice of $r$, the probability that $\mathcal{N}_{\mathbb{Z}}(0, \sigma^2)$ is less than $r$ is at least $1-\gamma/d = O(1)$. Furthermore, sampling from $\mathcal{N}_{\mathbb{Z}}(0, \sigma^2)$ takes at most $\polylog (\sigma^2 ) = \polylog (\frac{d \ln (1/\delta)} {\varepsilon})$ bits in expectation by \cref{lem: gauss_sample}. \cref{lem: rejection_sampling} then tells us that $\mathbb{E}[X_i | \overline{E}_1] = \polylog (\frac{d}{\varepsilon} \log(\frac{1}{\delta}))) = \polylog \frac{d}{\varepsilon} + \polylog \log \frac{1}{\delta}$. This allows us to bound $\mathbb{E}[X_i]$.
    \[ \mathbb{E}[X_i] = O \left( \frac{\polylog \frac{d}{\varepsilon} + \polylog \log \frac{1}{\delta}}{s} \right)\]

    Linearity of expectation tells us that the total number of random bits used over all loop iterations in expectation is the following.
    \[ O \left( \frac{d \left( \polylog \frac{d}{\varepsilon} + \polylog \log \frac{1}{\delta} \right)}{s} \right)\]

    Finally, the sampling of $\omega$ clearly takes $O(\log s)$ bits in expectation. We sample $\lceil \log s\rceil$ bits and reject until those bits represent a number between $1$ and $s$. This requires $O(\log s)$ bits in expectation by \cref{lem: rejection_sampling}. Therefore, by linearity of expectation, the total number of bits used by \cref{alg: approx_mec4} in expectation is the following.
    \[ O \left( \frac{d \left( \polylog \frac{d}{\varepsilon} + \polylog \log \frac{1}{\delta} \right)}{s} + \log(s) \right)\]

    \noindent \textbf{Time complexity:} Let $Y_i$ denote the time required to execute loop iteration $i$.
    \[ \mathbb{E}[Y_i] = \mathbb{E}[Y_i | E_i] \cdot \mathbb{P}[E_i] + \mathbb{E}[Y_i | \overline{E}_i] \cdot \mathbb{P}[\overline{E}_i]
    \leq \mathbb{E}[Y_i | E_i] + \mathbb{E}[Y_i | \overline{E}_i]\]
    $\mathbb{E}[Y_i | E_i]$ is clearly $\poly(n, d, \log(s))$. Furthermore, we already established above that in the case of $\overline{E}_i$, we only need 2 samples in expectation to sample $\eta_i$ conditioned on $|\eta_i| < r$. By \cref{lem: rejection_sampling}, that means the $\mathbb{E}[Y_i | \overline{E}_i] = \poly(n, d, \log s) \cdot \polylog(\sigma^2) = \poly(n, d, \log(s)) \cdot (\log \frac{d}{\varepsilon} + \log \log \frac{1}{\delta} )$. Combined with the above bound on $\mathbb{E}[Y_i]$, we get the following.
    \[ \mathbb{E}[Y_i] \leq  \poly(n, d, \log s) \left( \log \frac{d}{\varepsilon} + \log \log \frac{1}{\delta} \right) \]

    Linearity of expectation tells us that the total time used over all loop iterations in expectation is also
    \[ \poly(n, d, \log s) \left( \log \frac{d}{\varepsilon} + \log \log \frac{1}{\delta} \right). \]

    Finally, we can account for the running time of sampling $\omega$ as an extra additive $\log(rs)$ factor. This is subsumed by the running time spent in the loop. Therefore, the expected running time can be bounded by the following.
    \[ \poly(n, d, \log s) \left( \log \frac{d}{\varepsilon} + \log \log \frac{1}{\delta} \right) = \poly \left(n, d, \log s, \log \frac{1}{\varepsilon}, \log \log \frac{1}{\delta} \right). \]
\end{proof}

\subsection{Accuracy}
\begin{lemma}
    \label{lem: approx_accuracy}
    Assume $\delta \leq e^{-\varepsilon/2}$. Let $\gamma = \frac{1}{2} \delta/(e^\varepsilon + 1)$, $\sigma^2 = 4d\ln(2/\delta)/\varepsilon^2$ and $r = \sigma \sqrt{2 \ln(d) \cdot \ln (1/\gamma)}$. \cref{alg: approx_mec4} is $(\alpha, 0)$ accurate for the following parameter ranges.
    \[ \alpha = O \left( \frac{\sqrt{d \ln(d) \cdot \ln \frac{1}{\delta}} \cdot s}{\varepsilon} \right) \]
\end{lemma}
\begin{proof}
    We will show this theorem for \cref{alg: approx_mec3} instead of \cref{alg: approx_mec4} and the theorem will follow from the fact that \cref{alg: approx_mec3} and \cref{alg: approx_mec4} have the same exact output distribution (\cref{lem: same_dist_approx}). Since $\| \eta \|_{\infty} < r$ by definition, $\| \omega \|_\infty \leq rs$, and rounding down to the nearest multiple of $rs$ changes a value by at most $rs$, the output of \cref{alg: approx_mec3} is within $r + rs + rs = r (2s+1)$ of $\text{sum}(x)$ with probability 1.

    \[ r (2s+1) \leq 3rs = 3s \sqrt{\frac{8d}{\varepsilon^2} \ln \left( \frac{2}{\delta} \right) \cdot  \ln(d) \cdot \ln(1/\gamma)}
    = O \left( \frac{\sqrt{d \ln(d) \cdot \ln \frac{1}{\delta}} \cdot s}{\varepsilon} \right)\]
\end{proof}

\snote{Say you will not bother with analysis the runtime beyond polynomial since the running time is dominated by $n$ whenever the schemes are useful.}

\section{Pure differential privacy}
\label{sec: pure}
As in the case of approximate differential privacy, we first state our main result, then give the mechanism, and then prove its privacy, randomness complexity, time complexity, and accuracy separately.
\PureFinal
\begin{proof}
    Combine \cref{lem: pure_privacy}, \cref{lem: pure_randomness_complexity}, \cref{lem: pure_time_complexity}, and \cref{lem: pure_accuracy}.
\end{proof}

\subsection{The algorithm}
As in \cref{sec: approx}, for the sake of analysis, we give a series of algorithms that all have the desired differential
privacy guarantee, though only \cref{alg: mec5} is the one which we will show has low randomness and time complexity.
\begin{algorithm}
\caption{Mechanism 1}
\label{alg: mec1}
\begin{algorithmic}
\Require $x \in (\{ 0, 1\}^d)^n$, $\varepsilon$

\State $\eta \sim  \text{Lap}_{\mathbb{Z}}(d/\varepsilon)^d$

\noindent \Return $y \leftarrow \text{sum}(x)+\eta$

\end{algorithmic}
\end{algorithm}

\begin{algorithm}
\caption{Mechanism 2}
\label{alg: mec2}
\begin{algorithmic}
\Require $x \in (\{ 0, 1\}^d)^n$, $\varepsilon$, and integers $m, s>0$

\State $\omega \sim U[\{ 1, \dots, s \}] \cdot m$

\State $\eta \sim  \text{Lap}_{\mathbb{Z}}(d/\varepsilon)^d$

\State $y \leftarrow \text{sum}(x)+ (\omega, \dots, \omega) + \eta$

\noindent \Return $\lfloor y \rfloor_{ms}$

\end{algorithmic}
\end{algorithm}

\begin{algorithm}
\caption{Mechanism 3}
\label{alg: mec3}
\begin{algorithmic}
\Require $x \in (\{ 0, 1\}^d)^n$, $\varepsilon$, and integers $m, s>0$

\State $\omega \sim U[\{ 1, \dots, s \}] \cdot m$

\For {$i \in [1, d]$}

\State $\eta_i \sim  \text{Lap}_{\mathbb{Z}}(d/\varepsilon)$

\State $y_i \leftarrow \lfloor \text{sum}(x)_i+ \omega + \eta_i \rfloor_{ms}$

\EndFor

\noindent \Return $(y_1, \dots, y_d)$

\end{algorithmic}
\end{algorithm}

\begin{algorithm}
\caption{Mechanism 4}
\label{alg: mec4}
\begin{algorithmic}
\Require $x \in (\{ 0, 1\}^d)^n$, $\varepsilon$, and integers $m, s>0$

\State $p \leftarrow 2e^{-(\varepsilon/d)(m-1)}/(e^{\varepsilon/d} +1)$

\State Sample $J \subseteq [1, d]$ such that each $i \in [1, d]$ is in $J$ independently with probability $p$

\For {$i \in [1, d]$}

\If {$i \in J$}

    Sample $\eta_i$ from $\text{Lap}_{\mathbb{Z}}(d/\varepsilon)$ conditioned on $|\eta_i| \geq m$

\Else

    Resample $\eta_i$ from $\text{Lap}_{\mathbb{Z}}(d/\varepsilon)$ until $|\eta_i| < m$
    
\EndIf

\State $y_i \leftarrow \lfloor \text{sum}(x)_i+ \omega + \eta_i \rfloor_{ms}$

\EndFor

\noindent \Return $(y_1, \dots, y_d)$

\end{algorithmic}
\end{algorithm}

\begin{algorithm}
\caption{Mechanism 5}
\label{alg: mec5}
\begin{algorithmic}
\Require $x \in (\{ 0, 1\}^d)^n$, $\varepsilon$, and integers $m, s>0$

\State $p \leftarrow 2e^{-(\varepsilon/d)(m-1)}/(e^{\varepsilon/d} +1)$

\State $t \sim \text{Bin}(d, p)$

\State $J \sim U\left[ {d \choose [t]} \right]$ \Comment{Continually pick from $[d] \setminus J$}
\State $\omega \sim U[\{ 1, \dots, s \}] \cdot m$

\For {$i \in [1, d]$}

\If {$i \in J$}

    Sample $\eta_i$ from $\text{Lap}_{\mathbb{Z}}(d/\varepsilon)$ conditioned on $|\eta_i| \geq m$

    \State $y_i \leftarrow \lfloor \text{sum}(x)_i+ \omega + \eta_i \rfloor_{ms}$
    
\Else
    \If {$\lfloor \text{sum}(x)_i + \omega - m \rfloor_{ms} = \lfloor \text{sum}(x)_i + \omega + m \rfloor_{ms}$}
    
        \State $y_i \leftarrow \lfloor \text{sum}(x)_i + \omega -m  \rfloor_{ms}$
        
    \Else
    
       \State Resample $\eta_i$ from $\text{Lap}_{\mathbb{Z}}(d/\varepsilon)$ until $|\eta_i| < m$

       \State $y_i \leftarrow \lfloor \text{sum}(x)_i+ \omega + \eta_i \rfloor_{ms}$
       
    \EndIf
\EndIf

\EndFor

\noindent \Return $(y_1, \dots, y_d)$

\end{algorithmic}
\end{algorithm}

\begin{lemma}
    \label{lem: same_dist}
    For all $\varepsilon, m, s$, \cref{alg: mec2} and \cref{alg: mec5} have the same output distribution.
\end{lemma}
\begin{proof}
    We will show this by showing that \cref{alg: mec2}, \cref{alg: mec3}, \cref{alg: mec4}, and \cref{alg: mec5} all have the same output distributions.
    \begin{enumerate}
        \item \cref{alg: mec3} is simply unrolling the loop implicit in \cref{alg: mec2}. Therefore, the distribution of \cref{alg: mec3} is clearly the same as the distribution of \cref{alg: mec2}.

        \item Notice that in both \cref{alg: mec3} and \cref{alg: mec4}, $\eta_i$ and $\eta_j$ are independent. Therefore, it suffices to show that for any $i$, the distribution of $\eta_i$ in \cref{alg: mec3} and \cref{alg: mec4} are the same, since the output of the algorithms conditioned on any $\eta_1, \dots, \eta_d$ is the same. The probability that $\eta_i = k$ for any $k \in \mathbb{Z}$ in \cref{alg: mec3} is the following.
        \[ \frac{e^{\varepsilon/d} - 1}{e^{\varepsilon/d} + 1} e^{-\varepsilon|k|/d}\]

        The probability that $\eta_i = k$ for any $k \in \mathbb{Z}$ in \cref{alg: mec4} is the following. Let $Y \sim \text{Lap}_{\mathbb{Z}}(d/\varepsilon)$.
        \begin{align*}
            \mathbb{P}[\eta_i = k] &= \mathbb{P}[\eta_i = k | i \in J] \cdot \mathbb{P}[i \in J] + \mathbb{P}[\eta_i = k | i \notin J] \cdot \mathbb{P}[i \notin J]\\
            &= \mathbb{P}[\eta_i = k | i \in J] \cdot p + \mathbb{P}[\eta_i = k | i \notin J] \cdot (1-p)\\
            &= \mathbb{P}[Y = k | |Y| \geq m] \cdot p + \mathbb{P}[Y = k | |Y| < m] \cdot (1-p)\\
        \end{align*}
        There are two cases to consider. The first being when $k \geq m$. In which case, the above simplifies to the following.
        
        \begin{align*}
            \mathbb{P}[Y = k | |Y| \geq m] \cdot p &=
            \frac{\mathbb{P}[Y = k]}{\mathbb{P}[|Y| \geq m]} \cdot p =  \mathbb{P}[Y=k]
        \end{align*}
        The final equality follows from the fact that $p = \mathbb{P}[|Y| \geq m]$ (\cref{lem: discrete_lap_tail}). The second case is when $k \geq m$. In which case, the $\mathbb{P}[\eta_i = k]$ simplifies to the following.
        \begin{align*}
            \mathbb{P}[Y = k | |Y| < m] \cdot (1-p)
            &= \frac{\mathbb{P}[Y = k]}{\mathbb{P}[|Y| < m]} \cdot (1-p)\\
            &= \frac{\mathbb{P}[Y = k]}{(1 - \mathbb{P}[|Y| \geq m])} \cdot (1-p)\\
            &= \frac{\mathbb{P}[Y = k]}{(1 - p)} \cdot (1-p)\\
            &= \mathbb{P}[Y = k]
        \end{align*}
        Therefore, $\mathbb{P}[\eta_i = k] = \mathbb{P}[Y = k]$, as desired.

        \item Note that the distribution of $J$ did not change at all from \cref{alg: mec4} and \cref{alg: mec5}. In particular, the probability that $J$ equals any particular set of size $k$ remains
        \[ \frac{1}{{d \choose k}} \cdot \mathbb{P}[\text{Bin}(d, p) = k] = \frac{{d \choose k} p^k (1-p)^{d-k}}{{d \choose k}} = p^k (1-p)^{d-k} .\]
        Therefore, we will assume that in \cref{alg: mec5}, we sample $J \subseteq [1, d]$ such that each $i \in [1, d]$ is in $J$ independently with probability $p$ (the same as \cref{alg: mec4}). It suffices to show that conditioned on $J, \omega$, the output distributions of $(y_1, \dots, y_d)$ are the same. Fix $J, \omega$ and consider the distribution of outputs of \cref{alg: mec4} and \cref{alg: mec5} conditioned on $J, \omega$, call these $\mathcal{D}_4, \mathcal{D}_5$. Notice that in both $\mathcal{D}_4$ and $\mathcal{D}_5$, each coordinate is distributed independently, so we simply need to show that the distributions of coordinate $i$ of $\mathcal{D}_4$ and coordinate $i$ of $\mathcal{D}_5$, call these $X_i, Y_i$ have the same distribution. Note that if $i \in J$, $X_i$ and $Y_i$ have the same distribution by definition. If $i \in J$ and $\lfloor \text{sum}(x)_i + \omega -m  \rfloor_{ms} = \lfloor \text{sum}(x)_i + \omega + m \rfloor_{ms}$, then the distributions of $X_i$ and $Y_i$ are clearly the same since $X_i = \lfloor \text{sum}(x)_i + \omega -m  \rfloor_{ms} = Y_i$.  If $i \in J$ and $\lfloor \text{sum}(x)_i + \omega -m  \rfloor_{ms} \neq \lfloor \text{sum}(x)_i + \omega + m \rfloor_{ms}$, then the distributions $X_i, Y_i$ are clearly the same by definition.
    \end{enumerate}
\end{proof}

\subsection{Privacy}

\begin{lemma}
    \label{lem: pure_privacy}
    \cref{alg: mec5} is $(\varepsilon, 0)$ differentially private for any input values of $s, m$.
\end{lemma}
\begin{proof}
    We will show that since \cref{alg: mec1} is $(\varepsilon, 0)$ differentially private, \cref{alg: mec5} is $(\varepsilon, 0)$ differentially private.
    \begin{enumerate}
        \item \cref{alg: mec1} is $(\varepsilon, 0)$ differentially private \cref{lem: discrete_lap_private}. 

        \item Notice that \cref{alg: mec2} can be seen as taking the output of \cref{alg: mec1}, adding $(\omega, \dots, \omega)$ to it where $\omega \sim U[\{ 1, \dots, s\}] \cdot m$, and then rounding each coordinate down to a multiple of $ms$. Therefore, \cref{lem: postprocessing} and the fact that \cref{alg: mec1} is $(\varepsilon, \delta)$-differentially private tell us \cref{alg: mec2} is $(\varepsilon, \delta)$-differentially private.

        \item Since \cref{alg: mec2} and \cref{alg: mec5} have the same output distribution by \cref{lem: same_dist}, and \cref{alg: mec2} is $(\varepsilon, 0)$ differentially private by the above, \cref{alg: mec5} is $(\varepsilon, 0)$ differentially private.
    \end{enumerate}
\end{proof}

\subsection{Randomness complexity}
One crucial step in \cref{alg: mec5} we have left underspecified up till now is how to sample $\eta_i$ conditioned on $|\eta_i|$ being greater than $m$. We show that this is not too hard (it comes down to sampling from a geometric distribution), and then use that to bound the randomness complexity and time complexity of \cref{alg: mec5}.
\begin{lemma}
    \label{lem: samp_rounded_lap}
    Say $t$ has a $\polylog(t)$ bit binary expansion. Let $X$ be $\text{Lap}_{\mathbb{Z}}(t)$ conditioned on being more than distance $T$ from $0$. We can sample from $X$ in expected time $O(t \cdot \polylog(t))$ using expected randomness $O(t \cdot \polylog(t))$.
\end{lemma}
\begin{proof}
    Since the discrete Laplacian is symmetric, we will assume without loss of generality that we wish to sample from the case where $X \geq T$. One can flip a coin to decide if we are in that case or the case where $X \leq -T$. Let $p=1-e^{-1/t}$. Recall from \cref{lem: discrete_lap_dist} that one way to generate $\text{Lap}_{\mathbb{Z}}(t)$ is by sampling $U, V \sim \text{Ber}(p)$ and outputting $U-V$. We now wish to sample $U, V \sim \text{Ber}(p)$ conditioned on $U-V \geq T$. Say we knew the marginal distribution of $V$, then we would first sample $v \sim V$, sample $U$ conditioned on $V=v$ and $U-V \geq T$, and then output $U-V$. So consider the marginal distribution of $U$ conditioned on $V=v$ and $U-V \geq T$.
    \[ \mathbb{P}[U=u | V=v, U-V \geq T] = \mathbb{P}[U=u | U \geq T+v]\]
    By the memorylessness property of the geometric distribution, $U \sim \text{Geo}(p)+T+v$. Therefore, regardless of what $v$ is, $U-V$ conditioned on $U-V \geq T$ is just $\text{Geo}(p)+T$. So we need is to generate $W \sim \text{Geo}(p)$ and output $W+T$. We generate $W$ by flipping a coin $\text{Ber}(1-e^{-1/t})$ until we see $1$. By \cref{lem: ber_exp_sampling}, each one of these iterations takes $\polylog(t)$ time and since the number of iterations is geometrically distributed with mean $1/p = 1/(1-e^{-1/t}) = 1/(O(1/t)) = O(t)$, the generation process should take expected time $O(t \cdot \polylog(t))$. The same reasoning holds for the randomness complexity.
\end{proof}

\begin{lemma}
    \label{lem: outside_loop_randomness}
    All operations outside of the for loop in \cref{alg: mec5} require expected randomness $O(\log(d) +pd \log(d) + \log(s))$.
\end{lemma}
\begin{proof}
    There are three operations to consider. Let $X_1$ be the number of bits used to sample $t$, $X_2$ be the number of bits used to sample $J$ and $X_3$ be the number of bits used to sample $s$. We wish to find $\mathbb{E}[X_1 + X_2 + X_3] = \mathbb{E}[X_1] + \mathbb{E}[X_2] + \mathbb{E}[X_3]$. We bound these separately.
    \begin{enumerate}
        \item 
        $\mathbb{E}[X_1] \leq \log(d)$ which follows from \cref{thm: samp_point_mass} and the fact that if $X \sim \text{Bin}(d, p)$, then $H(X) = O(\log_2(d))$.
        
        \item 
        To sample $J$, we set $J \leftarrow \{ \}$ and then we continually pick from $[d] \setminus J$ until $|J| = t$. Sampling a random value from $[d] \setminus J$ can be done by sampling $\log_2(d-|J|)$ bits, checking if they correspond to a value in at most $d \setminus J$, if they do, then adding that value to $J$, and otherwise resampling if they do not. This process requires at most $2\log(d)$ bits in expectation. For any fixed $t$, sampling the $i^{\text{th}}$ value of $J$ requires $2\log(d)$ bits by \cref{lem: rejection_sampling}. Let $T$ be the value of $t$. 
        \[ \mathbb{E}[X_2] 
        = \mathbb{E}_T[\mathbb{E}[X_2|T]]
        \leq \mathbb{E}_T[\mathbb{E}[2 \log(d) T] | T]
        = 2 \log(d) \mathbb{E}_T[\mathbb{E}[T] | T]
        = 2 \log(d) pd\]

        \item 
        $X_3$ can clearly be sampled using at most $O(\log s)$ bits in expectation.
    \end{enumerate}
\end{proof}

\begin{lemma}
    \label{lem: pure_randomness_complexity}
    Assume $m = \frac{d}{\varepsilon} \ln(\frac{d}{\varepsilon}) \ln(s) + 1$ and $d/\varepsilon > 10$. \cref{alg: mec5} requires $O(\frac{d}{s} \cdot \polylog \frac{d}{\varepsilon} + \log s)$ bits of randomness in expectation.
\end{lemma}
\begin{proof}
    Let $X_i$ denote the expected number of bits used in iteration $i$ of the loop. We define $E_1, E_2, E_3$ as the following events respectively.
    \[ i \in J\]
    \[ (i \notin J) \cap (\lfloor \text{sum}(x)_i + \omega - m \rfloor_{ms} = \lfloor \text{sum}(x)_i + \omega + m \rfloor_{ms}) \]
    \[ (i \notin J) \cap (\lfloor \text{sum}(x)_i + \omega - m \rfloor_{ms} \neq \lfloor \text{sum}(x)_i + \omega + m \rfloor_{ms}) \]

    We will bound $\mathbb{E}[X_i]$ by bounding the conditional expectations one by one.
    \[ \mathbb{E}[X_i] = \sum_{j=1}^3 \mathbb{E}[X_i | E_j] \cdot \mathbb{P}[E_j] \]
    \begin{enumerate}
        \item $\mathbb{E}[X_i | E_1] \cdot \mathbb{P}[E_1]$. $\mathbb{E}[X_1 | E_1]$ is simply the number of random bits it takes in expectation to sample from $\text{Lap}_{\mathbb{Z}}(d/\varepsilon)$ conditioned on it being at least $m$. \cref{lem: samp_rounded_lap} tells us $\mathbb{E}[X_1 | E_1] = O(\frac{d}{\varepsilon} \polylog(\frac{d}{\varepsilon}))$. \cref{lem: discrete_lap_tail} tells us that $\mathbb{P}[E_1] \leq e^{-\varepsilon (m-1)/d}$. $\mathbb{E}[X_i | E_1] \cdot \mathbb{P}[E_1]$is therefore upper bounded by $O(\frac{d  e^{-\varepsilon (m-1)/d}}{\varepsilon} \polylog(\frac{d}{\varepsilon})) = O(\frac{1}{s} \cdot \polylog \frac{d}{\varepsilon})$.

        \item $\mathbb{E}[X_i | E_2] \cdot \mathbb{P}[E_2]$. If $E_2$ happens, the loop uses no randomness. Therefore, $\mathbb{E}[X_i | E_2] \cdot \mathbb{P}[E_2]$ is equal to 0.

        \item $\mathbb{E}[X_i | E_3] \cdot \mathbb{P}[E_3]$. First consider $\mathbb{E}[X_1 | E_3]$, the number of bits it takes to rejection sample from $\eta_i$ conditioned on $|\eta_i| < m$. The expected number of bits required for one sample from $\text{Lap}_{\mathbb{Z}}(d/\varepsilon)$ is $\polylog( \frac{d}{\varepsilon})$ by \cref{lem: lap_sample}. The probability that $|\eta_i| < m$ is the following by \cref{lem: discrete_lap_tail}.
        \[ 1 - \frac{e^{-\varepsilon (m-1)/d}}{e^{\varepsilon/d}+1}\]
        Therefore, by \cref{lem: rejection_sampling}, the following holds.
        \[ \mathbb{E}[X_i | E_3] = \frac{\polylog \frac{d}{\varepsilon}}{1 - \frac{e^{-\varepsilon (m-1)/d}}{e^{\varepsilon/d}+1}}\]
        Notice that as $d/\varepsilon \geq 10$, the denominator is bounded by a constant. Therefore, $\mathbb{E}[X_i | E_3] = \polylog \frac{d}{\varepsilon}$.
        $\mathbb{P}[E_3] = 2(1-p)/s$ by \cref{lem: 1/s_prob}. Therefore, the whole probability is upper bounded by $O( \frac{1}{s} \cdot \polylog \frac{d}{\varepsilon})$.
    \end{enumerate}

    Combining these facts tells us that $\mathbb{E}[X_i] = O(\frac{1}{s} \cdot \polylog(\frac{d}{\varepsilon}))$. The total number of bits used in the loop is the $X = X_1 + \dots + X_d$. $\mathbb{E}[X] = O(\frac{d}{s} \cdot \polylog(\frac{d}{\varepsilon}))$ by linearity of expectation.

    Finally, we must account for the randomness used before entering the loop. By \cref{lem: outside_loop_randomness}, this is $O(\log(d) + pd \log(d) + \log(s))$ in expectation. By our choice of $m$, $p \leq \varepsilon/d$. Therefore, the expected bits of randomness needed before the loop starts becomes $O(\log(d)+\log(s))$. Finally, by linearity of expectation, the total number of random bits needed for \cref{alg: mec5} in expectation is $O(\frac{d}{s} \cdot \polylog \frac{d}{\varepsilon}) + O(\log(d)+\log(s)) = O(\frac{d}{s} \cdot \polylog \frac{d}{\varepsilon} + \log s)$.
\end{proof}

\subsection{Time complexity}
\begin{lemma}
    \label{lem: pure_time_complexity}
    Assume $m = \frac{d}{\varepsilon} \ln(\frac{d}{\varepsilon}) \ln(s) + 1$ and $d/\varepsilon > 10$. \cref{alg: mec5} has expected running time $\poly(n, d, \log s, \frac{1}{\varepsilon})$.
\end{lemma}
\begin{proof}
    We first consider the running time of all operations which take place outside of the loop. We sample $t \sim \text{Bin}(d, p)$ using \cref{thm: samp_point_mass} with probabilities $p_i = {d \choose i} p^i (1-p)^{d-i}$. We claim without proof that each of these probabilities has a $\poly(d, \log s, \log \frac{1}{\varepsilon})$ bit description and can be calculated to within $2^{-t}$ using time $\poly(d, \log s, \log \frac{1}{\varepsilon}, t)$ time. This is simply a matter of looking at the Taylor expansion of $p$ and calculating it up to sufficient accuracy to ensure that the approximation of $p_i$ is within $2^{-t}$ of $p_i$. \cref{thm: samp_point_mass} therefore tells us that we can sample $t$ using $\poly(d, \log s, \frac{1}{\varepsilon})$ bits in expectation.

    The second operation of the loop can be written as the sum of $X_i$ where $X_i$ is the expected time it requires to draw the $i^{\text{th}}$ item of $J$. \cref{lem: rejection_sampling} tells $\mathbb{E}[X_i] = O(\log d)$. Linearity of expectation tells us that, the total time required for this step in expectation is $d \cdot \mathbb{E}[X_i] = O(d \log d)$.

    Sampling $\omega$ can clearly be done in $O(\log s)$ time.

    Now let $Y_i$ denote the expected running time of loop iteration $i$.
    \[ \mathbb{E}[Y_i] = \mathbb{E}[Y_i | i \in J] \cdot \mathbb{P}[i \in J] + \mathbb{E}[Y_i | i \notin J] \cdot \mathbb{P}[i \notin J]\]

    $\mathbb{E}[Y_i | i \in J]$ is some $T = O( \frac{d}{\varepsilon} \polylog(\frac{d}{\varepsilon}))$ by \cref{lem: samp_rounded_lap}. $\mathbb{E}[Y_i | i \notin J]$ is $\polylog(\frac{d}{\varepsilon})/\mathbb{P}[\text{Lap}_{\mathbb{Z}}(d/\varepsilon) < m]$ by \cref{lem: lap_sample} and \cref{lem: rejection_sampling}.

    \begin{align*}
        \mathbb{E}[Y_i] &= T \cdot p + \frac{\polylog(\frac{d}{\varepsilon})}{\mathbb{P}[\text{Lap}_{\mathbb{Z}}(d/\varepsilon) < m]} (1-p)\\
        &= T \cdot p + {\polylog \left( \frac{d}{\varepsilon} \right)}\\
        &= T \cdot \frac{2e^{-\varepsilon(m-1)/d}}{e^{\varepsilon/d}+1} + {\polylog \left( \frac{d}{\varepsilon} \right)}\\
        &\leq T \cdot 2e^{-\varepsilon(m-1)/d} + {\polylog \left( \frac{d}{\varepsilon} \right)}\\
        &\leq T \cdot \frac{\varepsilon}{d} + {\polylog \left( \frac{d}{\varepsilon} \right)}\\
        &\leq {\polylog \left( \frac{d}{\varepsilon} \right)}\\
    \end{align*}

    By linearity of expectation, the total running time spent in the loop in expectation is $\mathbb{E}[Y_1 + \dots + Y_d] = d \polylog(\frac{d}{\varepsilon})$. Finally, computing $\text{sum}(x)$ takes $\poly(n, d)$ time. Therefore, the total running time of \cref{alg: mec5} is $\poly(n, d, \log s, \log \frac{1}{\varepsilon
    })$ time in expectation.
\end{proof}

\subsection{Accuracy}
\begin{lemma}
    \label{lem: pure_accuracy}
    Let $\varepsilon$, $s$, and $m = \frac{d}{\varepsilon} \ln(\frac{d}{\varepsilon}) \ln(s) + 1$ be the input parameters to \cref{alg: mec5}. Then \cref{alg: mec5} is $(\alpha, \beta)$ accurate for the following $\alpha$.
    \[ \alpha = O\left( \frac{d \cdot \log(d/\beta)}{\varepsilon} + \frac{s \ln(s) \cdot d \ln(\frac{d}{\varepsilon})}{\varepsilon} \right) \]
\end{lemma}
\begin{proof}
    We will show that the theorem holds from \cref{alg: mec2} instead of \cref{alg: mec5} since \cref{alg: mec2} and \cref{alg: mec5} have the same output distribution by \cref{lem: same_dist}.
    Let $m' = \frac{d}{\varepsilon} \ln \frac{d}{\beta}$. A union bound over $\eta_i$ and \cref{lem: discrete_lap_tail} tell us that $|(\eta_1, \dots, \eta_d)|_\infty > m'$ with probability at most $\beta$. The addition of $(\omega, \dots, \omega)$ then changes each coordinate of $\text{sum}(x)+\eta$ by at most $sm$, the same applies for rounding down to the nearest multiple of $ms$. Therefore, with probability at most $\beta$,
    \[ | \lfloor \text{sum}(x) + \eta + (\omega, \dots, \omega) \rfloor_{ms} - \text{sum}(x)|_\infty >  \frac{d}{\varepsilon} \ln \frac{d}{\beta} + 2ms .\]
    Plugging in $m = \frac{d}{\varepsilon} \ln(\frac{d}{\varepsilon}) \ln(s) + 1$ yields the desired result.
\end{proof}

\appendix
\section{Sampling from efficiently approximable $p_1, \dots, p_k$}
The following theorem describes an algorithm to sample from a discrete probability distribution where each (possibly irrational) probability has an efficiently computable series of approximations. The algorithm proposed is essentially the same as those presented in the classic works on sampling \cite{knuth1976complexity, hoshi2002interval} where we are explicit about time complexity subtleties.
\begin{theorem}
    \label{thm: samp_point_mass}
    Let $p_1, \dots, p_k \in (0, 1)$ be probabilities summing up to 1 with $b$ bit descriptions such that approximating $p_i$ to within $2^{-t}$ can be done in $\poly(b, t)$. Let $X$ be the random variable which has weight $p_i$ on the integer $i$.
    There exists an algorithm which samples from $X$, runs in expected time $\poly(k, b)$, and uses the following number of random bits in expectation.
    \[ H(X) + O(1) = - \sum_{i=1}^k  p_i \log_2(p_i) + O(1)\]
\end{theorem}
\begin{proof}
    Let $P_0 = 0$, $P_i = \sum_{j=1}^i p_j$, and $Q_i^t$ be the approximation of $P_i$ such that $|P_i-Q_i| \leq 2^{-t}$. The algorithm is as follows.
    \begin{algorithm}
    \caption{SamplePointMass}
    \label{alg: sample_point_mass}
    \begin{algorithmic}
    \State $u \leftarrow 0$
    \For {$t$ from $1$ to $\infty$}
    \If {there exists $i$ such that $u+2^{-t} \leq Q_i^t - 2^{-t}$ and $u-2^{-t} \geq Q_{i-1}^t+2^{-t}$}
        \State \Return $i$
    \EndIf
    \State $u_t \leftarrow \text{Ber}(1/2)$
    \State $u \leftarrow u+2^{-t} \cdot u_t$
    
    \EndFor
    \end{algorithmic}
    \end{algorithm}

    We first show that \cref{alg: sample_point_mass} samples from $X$. Consider the measure preserving map from infinite bitstrings to $[0, 1]$, $\phi: (u_1, u_2, \dots) \mapsto \sum_{i=1}^\infty 2^{-i} \cdot u_i$. Note that the set of infinite bitstrings on which our algorithm terminates is exactly $\phi^{-1}((P_{i-1}, P_{i}))$ since it is on those sets that we will eventually find a $t$ such that we know $u$ with enough precision to say that $u+2^{-t} \leq Q_i^t - 2^{-t}$ and $u-2^{-t} \geq Q_{i-1}^t+2^{-t}$. Since $\phi$ is measure preserving, the measure of those bitstrings is exactly the measure of the interval $(P_{i-1}, P_{i})$, which is $p_i$.
     
     Next, we show that the the loop in \cref{alg: sample_point_mass} runs for $H(X)+O(1)$ iterations in expectation. This is equivalent to showing that \cref{alg: sample_point_mass} requires $H(X)+O(1)$ bits of randomness in expectation.
     

    Let $B$ be the number of iterations for which the algorithm runs. Let $r \in \{ 0, 1\}^{\infty}$ be the value on the randomness tape, $u = \phi(r)$, and $A$ be our algorithm. To bound $B$, we will condition on the output of the algorithm. Notice that $B > x$ implies $u + 2^{-x} > Q_i - 2^{-x}$ or $u - 2^{-x} < Q_{i-1} + 2^{-x}$. Respectively, this implies $u + 3 \cdot 2^{-x} > Q_i + 2^{-x} \geq P_i$ or $u - 3 \cdot 2^{-x} < Q_{i-1} - 2^{-x} \leq P_{i-1}$. So the probability that $B > x$ is upper bounded by the probability that $u + 3 \cdot 2^{-x} > P_i$ or $u - 3 \cdot 2^{-x} < P_{i-1}$, which is exactly $6 \cdot 2^{-x}$.
    \begin{align*}
        \mathbb{E}[B | A(r) = i] &= \sum_{x=0}^\infty \mathbb{P}[B > x | A(r) = i]\\
        &= \sum_{x=0}^{-\log_2(p_i)+2} \mathbb{P}[B > x | A(r) = i] + \sum_{x=-\log_2(p_i)+2}^{\infty} \mathbb{P}[B > x | A(r) = i]\\
        &\leq -\log_2(p_i)+3 + \sum_{x=-\log_2(p_i)+2}^{\infty} \mathbb{P}[B > x | A(r) = i]\\
        &\leq -\log_2(p_i)+3 + \sum_{x=-\log_2(p_i)+2}^{\infty} \frac{\mathbb{P}[B>x \cap A(r) = i]}{A(r) = i}\\
        &= -\log_2(p_i)+3 + \sum_{x=-\log_2(p_i)+2}^{\infty} \frac{6 \cdot 2^{-x}}{p_i}\\
        &= -\log_2(p_i)+ O(1)\\
    \end{align*}
    We can now use this to bound the expected number of iterations.    
    \begin{align*}
        \mathbb{E}[B] &= \sum_{i=1}^k \mathbb{P}[A(r) = i] \cdot \mathbb{E}[B | A(r) = i]\\
        &= \sum_{i=1}^k p_i \cdot (-\log_2(p_i) + O(1))\\
        &= \sum_{i=1}^k - p_i \cdot \log_2(p_i) + \sum_{i=1}^k O(1) \cdot p_i\\
        &= H(X) + O(1)
    \end{align*}

    Finally, to determine the expected number of operations, let $Y_t$ denote the expected number of operations performed in iteration $t$ of the loop ($Y_t=0$ if the loop has terminated by that point). Let $Y = \sum_{t=1}^\infty Y_t$. Again, we condition on the output of the algorithm.
    \begin{align*}
        \mathbb{E}[Y] &= \sum_{i=1}^k \mathbb{E}[Y|A(r)=i] \cdot \mathbb{P}[A(r)=i]\\
        &= \sum_{i=1}^k \mathbb{E}[Y|A(r)=i] \cdot p_i\\
        &=\sum_{i=1}^k \sum_{t=1}^\infty \mathbb{E}[Y_t|A(r)=i] \cdot p_i\\
        &=\sum_{i=1}^k \sum_{t=1}^\infty \mathbb{E}[Y_t|A(r)=i, Y_t \neq 0] \cdot \mathbb{P}[Y_t \neq 0 | A(r) = i] \cdot p_i
    \end{align*}

    Note that the only way that $Y_t \neq 0$ given that $A(r)=1$ is if the loop is still running at iteration $t$.  This is the probability that $u$ is within $O(2^{-t})$ of $P_{i-1}$ or $P_i$ conditioned on $u \in (P_{i-1}, P_i)$, which is $O(2^{-t})/p_i$. Let $c$ be some constant.
    \begin{align*}
        \mathbb{E}[Y] &=\sum_{i=1}^k \sum_{t=1}^\infty \mathbb{E}[Y_t|A(r)=i, Y_t \neq 0] \cdot \frac{O(2^{-t})}{p_i} \cdot p_i\\
        &=c \cdot \sum_{i=1}^k \sum_{t=1}^\infty \mathbb{E}[Y_t|A(r)=i, Y_t \neq 0] \cdot 2^{-t}\\
        &=c \cdot \sum_{i=1}^k \sum_{t=1}^\infty \poly(k, b, t) \cdot 2^{-t} \\
        &=c \cdot \sum_{i=1}^k \poly(k, b)\\
        &= \poly(k, b)
    \end{align*}
\end{proof}

\bibliographystyle{alpha}
\bibliography{main}

\end{document}